\documentclass{article}

\usepackage{amsthm,amsmath,amsfonts,amssymb}
\usepackage{mathrsfs,graphicx,bbm}
\usepackage[authoryear]{natbib}
\usepackage{authblk}

\theoremstyle{plain}

\newtheorem{theorem}{Theorem}[section]
\newtheorem{lemma}[theorem]{Lemma}
\newtheorem{corollary}[theorem]{Corollary}

\theoremstyle{remark}
\newtheorem{definition}[theorem]{Definition}

\newtheorem*{remark}{Remark}

\title{Making all pairwise comparisons in multi-arm clinical trials without control treatment}

\author[$\star$]{T. BURNETT}
\author[$\dagger\ddag$]{T. JAKI}

\affil[$\star$]{Department of Mathematical Sciences, University of Bath, Bath, BA2 7AY,  U.K.\\ 
tb292@bath.ac.uk}
\affil[$\dagger$]{MRC Biostatistics Unit, University of Cambridge, Cambridge, CB2 0SR, UK\\ thomas.jaki@mrc-bsu.cam.ac.uk}
\affil[$\ddag$]{Department of Machine Learning and Data Science, University of Regensburg, Regensburg, Germany}

\begin{document}

\maketitle

\begin{abstract}
The standard paradigm for confirmatory clinical trials is to compare experimental treatments with a control, for example the standard of care or a placebo. However, it is not always the case that a suitable control exists. Efficient statistical methodology is well studied in the setting of randomised controlled trials. This is not the case if one wishes to compare several experimental with no control arm. We propose hypothesis testing methods suitable for use in such a setting.  These methods are efficient, ensuring the error rate is controlled at exactly the desired rate with no conservatism. This in turn yields an improvement in power when compared with standard methods one might otherwise consider using, such as a Bonferroni adjustment. The proposed testing procedure is also highly flexible. We show how it may be extended for use in multi-stage adaptive trials, covering the majority of scenarios in which one might consider the use of such procedures in the clinical trials setting. With such a highly flexible nature, these methods may also be applied more broadly outside of a clinical trials setting.
\end{abstract}

\section{Introduction}

Confirmatory clinical trials, by their nature, require careful construction of the hypothesis tests to control the type $\mathrm{I}$ error rate \citep{ICHE9,bretz2020commentary}. A trial might consider multiple experimental treatments. For example, several doses of the same treatment, competing treatments or combinations of treatments. In such circumstances, we may incorporate multiple comparisons with a common control, while strongly controlling the Family-Wise Error Rate (FWER) \citep{magirr2012generalized,urach2016multi,stallard2003sequential,thall1989two,sampson2005drop}. As noted by \cite{bretz2020commentary} strong control of the FWER is required in the highly regulated environment of clinical trials \citep{ICHE9,FDA,EMA}. Such multi-arm methods are efficient for multiple comparisons within the same trial and can be extended to adaptive multi-stage designs \citep{magirr2012generalized,urach2016multi,stallard2003sequential,thall1989two,sampson2005drop}. Multi-Arm Multi-Stage designs open the door to the many benefits of adaptive designs \citep{pallmann2018, burnett2020adding}.

In other cases, there is no obvious control treatment. \cite{magaret2016design} (revisited in \cite{whitehead2020estimation}) consider a trial in Sepsis with multiple treatment options but no control arm. Suppose, for example, there is no standard of care and all treatments are used in clinical practice. There would be little motivation for patients to enrol to a trial where they may not receive treatment that they might access through not joining the trial. In such cases, the inclusion of a control treatment in the traditional sense may not be practical or ethical. Further, the control treatment may not be required to answer the scientific question of interest. For example, if we wish to directly compare all of the experimental treatments. In such situations, we suggest making all pairwise comparisons between each of the treatments in the trial. 

To strongly control the FWER while making all pairwise comparisons one might use a Bonferroni adjustment or a gatekeeping procedure \citep{garcia2008extension,jennison1999group,seaman1991new,cangur2014examining}. Similarly the Tukey's range test \citep{tukey1949} could be used in this context. We take inspiration from the approaches from tests with a common control, defining a closed testing procedure \citep{marcus1976closed} using a two-sided Dunnett type test \citep{dunnett1955multiple} for each intersection hypothesis. This procedure benefits from using the correlation structure that is introduced through the multiple pairwise comparisons, strongly controls the FWER as desired and increases the power of the test compared to the above alternatives. In addition, our hypothesis testing method may be extended for use in adaptive multi-stage designs. We discuss how to construct multi-stage tests using both group sequential \citep{jennison1999group,pocock1977group,o1979multiple,whitehead1997design,gordon1983discrete} and fully flexible \citep{bretz2006confirmatory,schmidli2006confirmatory} methods, covering all scenarios in which the test may be used in the setting of clinical trials.

Our focus is the setting of confirmatory clinical trials and hence the manuscript is written with this context in mind. The methods we introduce are not limited to this application however, the hypothesis testing structure is applicable in any experimental design where all pairwise comparisons are desired while strongly controlling the FWER. For example, in the comparison of advertising strategies or the agricultural sector comparing several fertilisers. While similar questions have been asked in the settings of psychology \citep{seaman1991new} and machine learning \citep{garcia2008extension}.

\section{All pairwise treatment comparisons in multi-arm trials}

\subsection{Multi-arm trial}
\label{sec:MA}

Suppose there are $K > 2$ candidate treatments $T_1,...,T_K$ and we wish to compare their comparative effectiveness. For each treatment $T_1,...,T_K$ we consider the corresponding treatment effect $\mu_1,...,\mu_K$. With no appropriate control we consider all pairwise comparisons, conducting formal analysis through the definition of corresponding hypotheses. For two treatments $i\neq j$ ($i,j=,1...,K$) we test the null hypothesis $H_{0,i,j}:\mu_i = \mu_j$ vs the alternative $H_{1,i,j}:\mu_i \neq \mu_j$. Swapping $i$ and $j$ will compare the same pair $\mu_i$ and $\mu_j$, so we only consider the tests that are unique. For $i=1,...,K-1$ we test the null $H_{0,i,j}$ for each $j$ such that $i < j \leq K$. Thus for $K$ treatments, we make $K(K-1)/2$ unique pairwise comparisons. We will refer to the setting were all means are identical, that is $\mu_i=\mu_j\ \forall\ i= 1, \dots, K$ and $j=1,\dots K$ as the global null hypothesis.

We plan a trial recruiting $n_i$ patients to treatment $T_i$ for $i=1,...,K$. Assuming that patient responses $x_{m,i}$ for $i=1,...,K$ and $m=1,...,n_i$ are independent and identically distributed following a normal distribution with known variance (this is not required but is suitable for demonstrating the methods to follow) we have that 
\begin{equation*}
x_{m,i} \sim N(\mu_i,\sigma^2_i)\text{ for } i=1,...,K\text{ and }m=1,...,n_i.
\end{equation*}
From which we estimate $\mu_i$ for $i=1,...,K$ $\mu_i$ by 
$\hat{\mu}_i = n_i^{-1}\Sigma_{m=1}^{n_i}x_{m,i}.$

We introduce an index for each of the unique pairwise comparisons. For the comparison of $T_i$ and $T_j$ where $i=1,...,K-1$ and $j$ such that $i < j \leq K$ we define the index $k$ as
\begin{equation*}
k = \begin{cases}
j-1 & \text{ if } i=1\\
j-1 + \Sigma_{c=1}^{i-1}K-c & \text{ if } i>1.
\end{cases}
\end{equation*}
We use this convention wherever $i,j,k$ appear throughout the manuscript. As an example, we write the null hypothesis $H_{0,i,j}$ as $H_{0,k}$.

For each $k=1,...,K(K-1)/2$ let 
\begin{equation*}
\sigma_{p,k}^2 = \frac{\sigma_i^2}{n_i} + \frac{\sigma_j^2}{n_j}\text{ and }\theta_k = \mu_i - \mu_j.
\end{equation*} 
We construct z-statistics for each pairwise comparison 
\begin{equation*}
Z_k = \frac{\hat{\theta}_k}{\sigma_{p,k}} \sim N(\theta_k\sigma^{-1}_{p,k},1)\text{ for }k=1,...,K(K-1)/2
\end{equation*}
so that $Z_k \sim N(0,1)$ under $H_{0,k}$. Supposing we wish to test $H_{0,k}$ at a type I error rate of $\alpha$, we reject $H_{0,k}$ when $|Z_k| > \Phi^{-1}(1-\alpha)$. Where $\Phi(.)$ is the standard normal cumulative distribution function.

\section{Testing procedure}

\subsection{Multiplicity adjustment}

Following Section~\ref{sec:MA} we can test all unique pairwise comparisons. However, in doing so we are making multiple comparisons with no adjustment. Given the desire to control the error rate it would be neglectful not to consider the impact of this multiplicity. Indeed the need for appropriate adjustment for multiplicity is well established in a clinical trials setting \citep{ICHE9,FDA,EMA}.
Strong control of the Family-Wise Error Rate (FWER) \citep{hochberg1987multiple,dmitrienko2009multiple} is a natural extension of control the type I error. 
For any $\mathcal{K} \subseteq \{1,...,K(K-1)/2\}$ suppose $k \in \mathcal{K}$ implies that $H_{0,k}$ is true. Defining $R_k$ to be the event that we reject $H_{0,k}$ for $k=1,...,K(K-1)/2$, strong control of the FWER at level $\alpha$ requires
\begin{equation*}
P(\bigcup_{k\in \mathcal{K}}R_k) \leq \alpha,
\end{equation*} 
for all $\mathcal{K} \subseteq \{1,...,K(K-1)/2\}$. 

To strongly control the FWER we require a closed testing procedure \citep{marcus1976closed}; whether this is implicitly or explicitly defined, as any procedure that strongly controls the FWER is equivalent to a closed testing procedure \citep{burnett2021adaptive}. To construct the closed testing procedure in addition to the individual pairwise tests we must also construct tests of all possible intersections of the null hypotheses. For $\mathcal{K} \subseteq \{1,...,K(K-1)/2\}$, we define the corresponding intersection hypothesis $H_{0,\mathcal{K}} = \bigcap_{k\in \mathcal{K}} H_{0,k}$. To reject $H_{0,k}$ (for $k=1,...,K(K-1)/2$) globally at level $\alpha$, we must reject local tests at level $\alpha$ for all null hypotheses in which it is involved. That is we reject $H_{0,k}$ globally if we reject all $H_{0,\mathcal{K}}$ for which $k\in\mathcal{K}$.

\subsection{Testing the intersection hypotheses}

\label{sec:tsdun}

We demonstrate how to construct the test for the global intersection $H_{0,F} = \bigcap_{k=1}^{K(K-1)/2} H_{0,k}$. The construction for all other intersection hypotheses follows identically using the corresponding joint distribution.

Under $H_{0,F}$ marginally $Z_k \sim N(0,1)$ for all $k=1,...,K(K-1)/2$. It follows that the joint distribution of $(Z_1,...,Z_{K(K-1)/2})$ is multivariate normal. This has a known correlation structure, introduced where one treatment appears in both sides of a pairwise comparison. For $k_1,k_2 \in \{1,...,K(K-1)/2\}$, if $k_1 = k_2$ it follows that $corr(Z_{k_1},Z_{k_2}) = 1$, while if $k_1 \neq k_2$ there are two cases to consider. If all four treatments are unique, that is for $k_1 = i_1,j_1$ and $k=i_2,j_2$ where $i_1,j_1,i_2,j_2 \in \{1,...,K(K-1)/2\}$ with $i_1 \neq i_2,j_2$ and $j_1 \neq i_2,j_2$, we have $corr(Z_{k_1},Z_{k_2}) = 0$. Alternatively, there is one shared treatment, for example, $k_1 = i,j_1$ and $k=i,j_2$, in which case,
\begin{equation*}
corr(Z_{k_1},Z_{k_2})  = \frac{\sigma^2_i/n_i}{\sigma_{p,k_1}\sigma_{p,k_2}}
\end{equation*}
Given the known joint distribution we construct a two-sided Dunnett type \citep{dunnett1955multiple} test to leverage this knowledge.

\begin{definition}
\label{def:dunint}
To test $H_{0,F}$ we define the test statistic
\begin{equation*}
Z_F = max(|Z_1|,...,|Z_{K(K-1)/2}|),
\end{equation*}
rejecting $H_{0,F}$ when $Z_F > C_{F,\alpha}$. Let $\theta_F = (\theta_1,...,\theta_{K(K-1)/2})$, under $H_{0,F}$  $\theta_F = \boldsymbol{0}$, we choose $C_{F,\alpha} > 0$ such that 
\begin{equation*}
P_{\theta_F = \boldsymbol{0}}(Z_F > C_{F,\alpha}) = \alpha.
\end{equation*}
\end{definition}

To compute $C_{F,\alpha}$ note that 
\begin{equation*}
P_{\theta_F = \boldsymbol{0}}(Z_F > C_{F,\alpha}) = 1-P_{\theta_F = \boldsymbol{0}}(\bigcap_{k=1}^{K(K-1)/2} -C_{F,\alpha} < Z_k < C_{F,\alpha}),
\end{equation*}
where $P_{\theta_F = \boldsymbol{0}}(\bigcap_{k=1}^{K(K-1)/2} -C_{F,\alpha} < Z_k < C_{F,\alpha})$ may be computed through standard application of the multivariate normal distribution. For example, by the multivariate normal package in R \citep{mvtnorm,compnorm}.

\begin{remark}
For each $\mathcal{K} \subseteq \{1,...,K(K-1)/2\}$ we can construct the test for the null hypotheses of the form $H_{0,\mathcal{K}} = \bigcap_{k\in\mathcal{K}}H_{0,k}$ at level $\alpha$ following the all pairwise Dunnett method given in Definition~\ref{def:dunint}. Applying each of these tests at level $\alpha$ as part of an overall closed testing procedure will ensure strong control of the FWER.
\end{remark}

\begin{lemma}
\label{lem:consonance}
If all intersection tests are constructed using the all pairwise Dunnett test the overall testing procedure is consonant \citep{romano2011consonance}. That is for $\mathcal{K}_1,\mathcal{K}_2 \subseteq \{1,...,K(K-1)/2\}$ with corresponding intersection hypotheses $H_{0\mathcal{K}_1} = \bigcap_{k\in\mathcal{K}_1}H_{0k}$ and  $H_{0\mathcal{K}_2} = \bigcap_{k\in\mathcal{K}_2}H_{0k}$. If $\mathcal{K}_1 \subset \mathcal{K}_2$ then $C_{\mathcal{K}_1,\alpha} < C_{\mathcal{K}_2,\alpha}$.
\end{lemma}

\begin{remark}
\label{rem:2}
The consequence of consonance is that if the global intersection hypothesis $H_{0,F}$ is rejected at least one null hypothesis is rejected globally. This is desirable as it ensures the overall procedure is constructed exactly at the desired error rate at every level. Furthermore, it is useful for designing the trial, as we may use the probability of rejecting $H_{0,F}$ as the measure of trial performance for setting the sample size.
\end{remark}

\subsection{Sample size calculation}
\label{sec:masize}

To plan trials we consider the probability of rejecting one or more null hypotheses (\citep[also known as the disjunctive power][]{vickerstaff2019methods}). As noted in Remark~\ref{rem:2}, this occurs whenever we reject $H_{0,F}$ as a consequence of Lemma~\ref{lem:consonance}. For a given configuration $\mu =\xi$ (where $\mu = (\mu_1,...,\mu_k)$) we choose the sample size to achieve $P_{\mu=\xi}(\text{Reject }H_{0,F}) = 1 - \beta$. Under any configuration $\mu=\xi$ the distribution of $Z_1,...,Z_{K(K-1)/2}$ remains multivariate normal with the correlation structure as described in Section~\ref{sec:tsdun}; however, 
\begin{equation*}
E(Z_k) = \frac{\mu_i-\mu_j}{\sigma_{p,k}}\text{ for any }k=1,...,K(K-1)/2.
\end{equation*}
We find $P_{\mu=\xi}(\text{Reject }H_{0,F})$ as described in Section~\ref{sec:tsdun},
\begin{equation*}
P_{\mu=\xi}(\text{Reject }H_{0,F}) = 1 - P_{\mu=\xi}(\bigcap_{k=1}^{K(K-1)/2} -C_{F,\alpha} < Z_k < C_{F,\alpha})
\end{equation*}
where $P_{\mu=\xi}(\bigcap_{k=1}^{K(K-1)/2} -C_{F,\alpha} < Z_k < C_{F,\alpha})$ may be computed through using standard multivariate normal functions in statistical software packages.\\

Let the total sample size be $n = \Sigma_{i=1}^{K}n_i$, then for a given $r_i$ where $r_in=n_i$ and known $\sigma_i$ (for $i=1,...,K$). We only need select $n$ to achieve $P_{\theta_F}(\text{Reject }H_{0,F}) = 1 - \beta$ under a target configuration of $\mu=\xi$. This can be achieved for example with a numerical search.

\subsection{Least favourable configuration}

\begin{definition}
\label{def:lfc}
Suppose $\delta$ is a clinically relevant difference between treatments that we wish to detect. Assuming this difference is present between two treatments, say $\mu_1 = \delta$ and $\mu_2=0$, the least favourable configuration is $\mu_i=\delta/2$ for $i=3,...,K$. We denote this by $\mu = \mu_{L}$
\end{definition}

\begin{theorem}
\label{the:lfc}
Assuming $\sigma_i^2/n_i = \sigma_j^2/n_j$ for all $i,j=1,...,K$, the least favourable configuration minimizes $P_\mu(\text{Reject }H_{0,1})$ for all $\mu$ with $\mu_1 = \delta$ and $\mu_2 = 0$. 
\end{theorem}

\begin{proof}
By Lemma~\ref{lem:consonance} we only need to consider $P_{\mu = \mu_{L}}(\text{Reject }H_{0,F})$. Considering the vector $Z_F = (Z_1,...,Z_{K(K-1)/2})$ we have that $E(Z_F) = \zeta_F = (\zeta_1,...,\zeta_{K(K-1)/2}) = (\theta_1/\sigma_{p,1},...,\theta_{K(K-1)/2}/\sigma_{p,K(K-1)/2})$ and denote the covariance matrix (which is as described in Section~\ref{sec:tsdun}) by $\Sigma$. Thus we have that $E(Z_k) = \theta_k$. The joint multivariate normal density function for $Z_F$ is given by,
\begin{equation*}
f_{\zeta_F}(Z_F) = 2\pi^{-k(K-1)/4}det(\Sigma)^{-1/2}e^{-1/2(Z_F-\zeta_F)\Sigma^{-1}(Z_F-\zeta_F)}.
\end{equation*}
Let $\mathcal{A}$ be the set of z-statistics where $|Z_k| < C_{F,\alpha}$ for all $k=1,...,K(K-1)/2$, then 
\begin{equation*}
P_\mu(\text{Reject } H_{0,F}) = \int_{Z_F\in\mathcal{\bar{A}}}f(Z_F)\mathrm{d}Z.
\end{equation*}

Note that for $k=1,...,K(K-1)/2$ $f(Z_F)$ is symmetric about the mean. That is $f_{\zeta_F}(Z_F+\zeta_F) = f_{\zeta_F}(-Z_F+\zeta_F)$. Also if $|\zeta_k - Z_k|$ decreases then $f_{\zeta_F}(Z_F)$ increases.  So for any constant $c>0$ if $\theta_k >0$ for all $k=1,...,K(K-1)/2$ 
we have that $f_{\zeta_F}(\boldsymbol{c}) > -f_{\zeta_F}(\boldsymbol{c})$. With the opposite being true for $c<0$ if $\theta_k < 0$. Thus $P_{\theta_F}(\text{Reject }H_{0,F})$ is minimised under $\theta_F=\boldsymbol{0}$. While for any $\theta_k + \epsilon > 0$ for $k=1,...,K(K-1)/2$ and $\epsilon > 0$ we have that $P_{\theta_F}(\text{Reject }H_{0,F})$ increases with $\epsilon$. While the equivalent is true for $\theta_k + \epsilon < 0$ and $\epsilon < 0$. 

Under the least favourable configuration assume without loss of generality that $\mu_1 = \delta$, $\mu_2=0$ and $\mu_i=\delta/2$ for $i=3,...,K$. We have that $\mu_1 - \mu_2 = \delta$ (for $\delta>0$), $\mu_1-\mu_i = \delta/2$, $\mu_2 - \mu_i = -\delta/2$ and $\mu_i-\mu_j = 0$ for all $i=1,...,K-1$ and $j$ such that $i < j \leq K$.

Consider a shift from the least favourable configuration, for some $\epsilon > 0$ let $\mu_3 = \delta/2 + \epsilon$. 

There is no change for the first two treatments with $\mu_1 - \mu_2 = \delta$ as before, and thus no impact on $P_{\theta_F}(\text{Reject }H_{0,F})$. 

For $j=4,...,K$,  $\mu_3-\mu_j = \epsilon$ which as described will increase $P_{\theta_F}(\text{Reject }H_{0,F})$ compared to the least favourable configuration. 

Finally, $\mu_1-\mu_3 = \delta/2 - \epsilon$, $\mu_2 - \mu_3 = -\delta/2 +\epsilon$, the impact of this on $(Z_F-\zeta_F)\Sigma^{-1}(Z_F-\zeta_F)$ are exactly opposed under the assumption that $\sigma_i^2/n_i = \sigma_j^2/n_j$ for all $i,j=1,...,K$. So $f(Z_F)$ is not changed by this and thus $P_{\theta_F}(\text{Reject }H_{0,F})$ is not impacted by this change from the least favourable configuration.

The equivalent to this is true for $\epsilon < 0$, and thus for any $\epsilon \neq 0$ a shift from the least favourable configuration increases $P_{\theta_F}(\text{Reject }H_{0,F})$. Clearly the same is also true if we apply a shift to any combination of $\mu_j$ for $j=3,...,K$. Suppose, for example, for $\epsilon_3,\epsilon_4 > 0$ with $\epsilon_3 \neq \epsilon_4$ we have $\mu_3 = \delta/2 + \epsilon_3$ and $\mu_=4 = \delta/2 + \epsilon_4$. All of the above holds while in addition $\mu_3 \neq \mu_4$ which would also increase $P_{\theta_F}(\text{Reject }H_{0,F})$
\end{proof}

\begin{remark}
Intuitively under the least favourable configuration $\mu_i$ sits at the midpoint of $\mu_1$ and $\mu_2$ for all $i=3,..,K$. So the probability of rejecting $H_{0,F}$ based on $Z_k > C_{F,\alpha}$ for comparisons with $i=1$ (which have mean proportional to $\delta/2$) and the probability of rejecting $H_{0,F}$ based on $Z_k < -C_{F,\alpha}$ for comparisons with $i=2$ (which have mean proportional to $-\delta/2$) are equal. 

Suppose the assumption that  $\sigma_i^2/n_i = \sigma_j^2/n_j$ for all $i,j=1,...,K$ is not true, we now have that $|E(Z_{k})| \neq E(Z_1)/2$ for all $k=2,...,2K-1$ and thus this even balance is not maintained. To restore this one might choose $\boldsymbol{a} = (a_1,...,a_K)$ such that $a\mu$ achieves the desired balance for $\mu_1,\mu_2$. However, $E(Z_k)$ may no longer be equal to $0$ for $k = 2K,...,K(K-1)/2$, which will increase $P_{\theta_F}(\text{Reject }H_{0,F})$. In such circumstances using $\boldsymbol{\mu} = (\delta,0,\delta/2,...,\delta/2)$ may remain a reasonable choice for a scenario that gives low probability of rejecting at least on null hypothesis. If a truly least favourable configuration is desired, a numerical search may be an appropriate way to find $\boldsymbol{a}$ that minimises $P_{\theta_F}(\text{Reject }H_{0,F})$.
\end{remark}

\subsection{One-sided hypothesis tests}

Our focus is on two-sided alternative hypotheses as this best reflects the intention when making all pairwise comparisons. For completeness, it is possible to consider one-sided alternatives using the same concepts. Consider null hypotheses of the form $H_{0,i,j}:\mu_i = \mu_j$ vs the alternative hypotheses $H_{1,i,j}:\mu_i >\mu_j$, for $i,j=1,...,K$ with $i\neq j$. This gives null hypotheses of the form  $H_{0k}:\theta_k \leq 0$ and alternatives $H_{1k}:\theta_k > 0$ for $k=1,...,K(K-1)$. For $i,j = 1,...,K$ in addition to the correlation we already discussed in Section~\ref{sec:tsdun}  we have that $corr{(\mu_i-\mu_j),(\mu_j-\mu_i)} = -1$. The test statistic for the one-sided all pairwise Dunnett test for $H_{0,F}$ is constructed as $max(Z_1,...,Z_{K(K-1)})$. 

\section{Comparison of performance}

\subsection{Bonferroni adjustment}

A Bonferroni correction \citep{bonferroni1936teoria,friedman2015fundamentals} strongly controls the FWER. For $H_{0,1},...,H_{0,K(K-1)/2}$ globally the hypothesis is rejected when $Z_k > \phi^{-1}[\alpha/\{K(K-1)/2\}]$ for all $k=1,...,K(K-1)/2$. This procedure implies a closed testing procedure. For any $\mathcal{K} \subseteq (1,...,K(K-1)/2)$ the intersection hypothesis $H_{0,\mathcal{K}} = \bigcap_{k\in\mathcal{K}} H_{0k}$ is implicitly rejected when $\bigcup_{k\in\mathcal{K}} (Z_k > \phi^{-1}[1-\alpha/\{K(K-1)/2\}])$. 

\begin{theorem}
\label{the:bonferonni}
The probability of globally rejecting $H_{0k}$ for $k=1,...,K(K-1)/2$ is higher for the all pairwise Dunnett testing procedure than for a Bonferroni adjusted testing procedure.
\end{theorem}

\begin{proof}
Consider the probability of rejecting the global intersection hypothesis $H_{0,F}$. Using the Bonferroni procedure under $H_{0,F}$ we have that
\begin{equation*}
\begin{split}
P_{\theta_F}(\text{Reject }H_{0,F}) &= P_{\theta_F}\{\bigcup_{k\in\mathcal{K}}(Z_k > \phi^{-1}[1-\alpha/\{K(K-1)/2\}])\}\\
&<\Sigma_{k\in\mathcal{K}}P\{(Z_k > \phi^{-1}[1-\alpha/\{K(K-1)/2\}])\}\\
&=\alpha.
\end{split}
\end{equation*}
So by construction of the all pairwise Dunnett procedure, we have that $C_{F,\alpha} < \phi^{-1}[1-\alpha/\{K(K-1)/2\}]$. Thus the probability of rejecting $H_{0,F}$ is higher for the all pairwise Dunnett test. 

By Lemma~\ref{lem:consonance} for any $\mathcal{K} \subseteq (1,...,K(K-1)/2)$, $C_{K,\alpha} < C_{F,\alpha} < \phi^{-1}[1-\alpha/\{K(K-1)/2\}]$. Thus for any local hypothesis test, there is a higher probability of rejection under the all pairwise hypothesis test, hence the probability of globally rejecting $H_{0k}$ for all $k=1,...,K(K-1)/2$ is higher for the all pairwise Dunnett testing procedure.
\end{proof}

\subsection{Gatekeeping procedure}

Gatekeeping procedures also offer strong control of the FWER. For such a procedure the hypotheses are ordered before the trial begins, then at the end of the trial each hypothesis is tested at the full $\alpha$ in order. If a null hypothesis is rejected then the next hypothesis is tested. If at any point there is a failure to reject a null hypothesis the procedure stops and we fail to reject all subsequent null hypotheses. 

\begin{remark}
Such a procedure is inherently reliant on the order in which the hypotheses are tested. Suppose, following the principle of clinical equipoise \citep{freedman2017equipoise} with no assumption about which treatment(s) are superior, we randomly order the hypotheses. Consider the least favourable configuration, with randomly order hypotheses there is only one order where the best treatment effect is the first hypothesis, while there are $(K-2)(K-3)/2$ true null hypotheses where by construction the probability of rejection is $\alpha$. Considering all possible combinations it should be clear that there is little motivation for this method without clear scientific motivation for the ordering (such as all treatments being dose levels of the same treatment). 
\end{remark}

\subsection{Tukey's range test}

Tukey's range test \citep{tukey1949} has some similarities to the proposed all pairwise Dunnett test. The two procedures are directly comparable in the case of the global null hypothesis, as defined in Section \ref{sec:MA}).

Tukey's range test uses the maximum difference between observed means to test the global null hypothesis. Let $\hat{\mu}_{min} = min(\hat{\mu}_1,...,\hat{\mu}_K)$ and $\hat{\mu}_{max} = max(\hat{\mu}_1,...,\hat{\mu}_K)$, supposing $\sigma_1^2=\sigma_2^2=,...,=\sigma_K^2$ and $n_1^2=n_2^2=,...,=n_K^2$ and denoting the pooled sample variance of $\mu_{min}$ and $\mu_{max}$ by $S^2$ the test statistic is given by 
\begin{equation*}
\frac{\hat{\mu}_{max}-\hat{\mu}_{min}}{S\sqrt{2/n}}
\end{equation*}
which follows the studentised range distribution with $nK - K$ degrees of freedom. Thus when these assumptions of equal sample size and variance are met the tests of the global null constructed using either the all pairwise Dunnett and Tukey's range test are asymptotically equivalent, spending the full $\alpha$ (note that in the scenario of unknown variance, one may still construct the all pairwise Dunnett test using the multivariate t-distribution).

The tests for other intersections are not defined for Tukey's range test. Consider the case where comparing two Z-values that have no individual arms in common it is not clear how one might construct this test using Tukey's method. For example, consider the test of a null hypothesis of the form $H_0: \mu_1 = \mu_2 \cap \mu_3 = \mu_4$, under the all pairwise Dunnett method this comes down to that values of $\hat{\mu}_1 - \hat{\mu}_2$ and $\hat{\mu}_3 - \hat{\mu}_4$ while considering the range through the use of $\hat{\mu}_{max} - \hat{\mu}_{min}$ does not properly define an equivalent test. Since most of the benefit of the all pairwise Dunnett procedure comes from fully defining the closed testing procedure, this presents an advantage over Tukey's range test, while also not requiring the distributional assumptions imposed by Tukey's test. We demonstrate clearly the advantage the fully defined closed testing procedure yields in Section~\ref{sec:sim}.

\subsection{Simulation study}
\label{sec:sim}

To evaluate the performance of the proposed all-pairwise Dunnett test we conduct a small simulation study to compare its performance to a Bonferroni adjustment and Tukey's range test (denoted Global). Gatekeeping is not included here due to the choice of order that is necessary for which prior scientific insight is necessary. We simulate $10^6$ hypothetical studies, testing 4 different treatments under three different set of true treatment means. The study is designed to control the FWER at 5\% and achieve a power of 90\% for a standardized treatment difference of 0.3743. The required sample size in this case is 809 patients per arm. \\

Table \ref{tab:sim} shows the estimated probability of rejection of at least 1 hypothesis (first column) as well as the probability of rejecting a number of hypotheses for different true mean vectors. As expected we find that all three approaches provide control of the family-wise error rate although slight conservatism can be seen for the Bonferroni method. This is in contrast to an unadjusted test which exhibits a 4-fold increase error.\\

When considering situations where some treatments have a different mean it is notable that all three approaches have almost identical power to reject at least one hypothesis. At the same time the number of rejections is consistently larger for the global test then Bonferroni which is dominated by the all-pairwise Dunnett test as expected. The difference is most notable in the last case where the probability to reject 4 hypothesis is at 11\% notably larger for the Dunnett test than for the global test.

\begin{table}[!ht]
\begin{tabular}{c c | c c c c c c c}
$\boldsymbol{\mu}$ & Test & Reject $\geq 1$ & 1 & 2 & 3 & 4 & 5 & 6\\\hline
(0,0,0,0) & Dunnett    & 0.05 & 0.04 & 0.01 & 0.00 & 0.00 & 0.00 & 0.00\\
(0,0,0,0) & Global     & 0.05 & 0.04 & 0.01 & 0.00 & 0.00 & 0.00 & 0.00\\
(0,0,0,0) & Bonferroni & 0.04 & 0.03 & 0.01 & 0.00 & 0.00 & 0.00 & 0.00\\
(0,0,0,0) & Unadjusted & 0.20 & 0.13 & 0.06 & 0.02 & 0.00 & 0.00 & 0.00\\
&&&&&&&\\
(10,5,5,0) & Dunnett    & 0.78 & 0.30 & 0.24 & 0.21 & 0.02 & 0.01 & 0.00\\
(10,5,5,0) & Global     & 0.78 & 0.33 & 0.27 & 0.17 & 0.01 & 0.00 & 0.00\\
(10,5,5,0) & Bonferroni & 0.76 & 0.34 & 0.26 & 0.15 & 0.01 & 0.00 & 0.00\\
&&&&&&&\\
(10,10,0,0) & Dunnett    & 0.96 & 0.05 & 0.17 & 0.19 & 0.55 & 0.01 & 0.00\\
(10,10,0,0) & Global     & 0.96 & 0.06 & 0.22 & 0.24 & 0.44 & 0.00 & 0.00\\
(10,10,0,0) & Bonferroni & 0.96 & 0.07 & 0.24 & 0.24 & 0.41 & 0.00 & 0.00\\
\end{tabular}
\caption{Probabilities of rejecting multiple hypotheses under each testing procedure under the Global Null and the LFC.\label{tab:sim}}
\end{table}

\section{Group sequential}

\subsection{Multi-stage recruitment and notation}
\label{sec:msrec}

The all pairwise Dunnett methodology may be extended to offer a suitable testing structure for multi-stage (adaptive) trials. Suppose we have $K > 2$ treatments for which we wish to test null hypotheses $H_{0,k}: \theta_k = 0$ vs alternative hypotheses $H_{1,k}:\theta_k \neq 0$ for $k = 1,...,K(K-1)/2$ as described in Section~\ref{sec:MA}. In the multi-stage setting, the trial is conducted over $Q >1$ stages, with an analysis after each stage of the trial. 

In this multi-stage setting at analysis $q=1,...,Q$ we have recruited $n^{(q)}_i$ patients to treatment $i=1,...,K$. We use the superscript $(q)$ to show a value is based on all patients recruited up to analysis $q = 1,...,Q$ of the trial, with everything else following the construction given in Section~\ref{sec:MA}. For example, $Z_k^{(q)}$ would be constructed using the corresponding $n_i^{(q)}$ and $n_j^{(q)}$ for all $k=1,...,K(K-1)/2$ and $q=1,...,Q$. The correlation structure follows for these z-statistics follows as in Section~\ref{sec:tsdun}. For $q=1,...,Q$ and $k_1,k_2 = 1,...,K(K-1)/2$ then if $k_1=k_2$, $corr(Z^{(q)}_{k_1},Z^{(q)}_{k_2}) = 1$, if $k_1 = i_1,j_1$ and $k=i_2,j_2$ where $i_1,j_1,i_2,j_2 \in (1,...,K(K-1)/2)$ with $i_1 \neq i_2,j_2$ and $j_1 \neq i_2,j_2$ then $corr(Z_{k_1},Z_{k_2}) = 0$, otherwise there is one shared treatment, for example $k_1 = i,j_1$ and $k=i,j_2$, in which case,
\begin{equation*}
corr(Z_{k_1},Z_{k_2}) = \frac{\sigma^2_i/n^{(q)}_i}{\sigma^{(q)}_{p,k_1}\sigma^{(q)}_{p,k_2}}.
\end{equation*}

\subsection{Defining the testing procedure}
\label{sec:mstest}

To construct the testing procedure in the multi-stage setting we follow the approach of \cite{urach2016multi}. Defining the group sequential testing boundary for each intersection hypothesis. As such we may construct an overall closed testing procedure and thus ensure strong control of the FWER. 

\begin{definition}
\label{def:gsdunint}
Let $\mathcal{K} \subset (1,...,K(K-1)/2)$, we define the corresponding intersection hypothesis $H_{0,\mathcal{K}} = \bigcap_{k\in \mathcal{K}} H_{0,k}$. For each stage $q=1,...,Q$ the test statistic for $H_{0,\mathcal{K}}$ is given by
\begin{equation*}
Z_{\mathcal{K}}^{(q)} = max(|Z_{k_1}^{(q)}|,|Z_{k_2}^{(q)}|,...),
\end{equation*}
where $(k_1,k_2,...)=\mathcal{K}$. These test statistics are sequentially compared to the testing boundaries $C_{\mathcal{K},\alpha} = (C_{\mathcal{K},\alpha}^{(1)},...,C_{\mathcal{K},\alpha}^{(q)})$, rejecting $H_{0,\mathcal{K}}$ (where $C_{\mathcal{K},\alpha}^{(q)}>0$ for all $q=1,...,Q$) at analysis $q$ if $Z_{\mathcal{K}}^{(q)} > C_{\mathcal{K},\alpha}^{(q)}$. Suppose that $H_{0,\mathcal{K}}$ is rejected at analysis $q^\star = 1,...,Q-1$, then no further analyses for $H_{0,\mathcal{K}}$ are conducted for any $q>q^\star$. If at analysis $q=1,...,Q-1$ $Z_{\mathcal{K}}^{(q)} < C_{\mathcal{K},\alpha}^{(q)}$, then the hypothesis will be tested again at analysis $q+1$. If at analysis $Q$ $Z_{\mathcal{K}}^{(q)} < C_{\mathcal{K},\alpha}^{(q)}$ we fail to reject $H_{0,\mathcal{K}}$.

To ensure $H_{0,\mathcal{K}}$ is tested at some pre-specified $\alpha$, $C_{\mathcal{K},\alpha}$ chosen such that
\begin{equation*}
P(\bigcup_{q=1}^{Q}\bigcup_{k\in\mathcal{K}} Z_{\mathcal{K}}^{(q)} > C_{\mathcal{K},\alpha}^{(q)}) = \alpha
\end{equation*}
\end{definition}

\begin{remark}
For each $\mathcal{K} \subseteq \{1,...,K(K-1)/2\}$ we may construct the test for $H_{0,\mathcal{K}}$ at level $\alpha$ using the group sequential all pairwise Dunnett test. Thus we can construct an overall closed testing procedure as defined in Section~\ref{sec:tsdun} and thus strongly control the FWER.
\end{remark}

Computation of $C_{\mathcal{K},\alpha}$ follows in a similar way to that set forth by \cite{magirr2012generalized}. 
We use the independent increments of the z-statistic at each stage trial to break them down by the observations in each stage. In the first stage of the trial $n'^{(1)}_i=n^{(1)}_i$ patients are recruited to each treatment $i=1,...,K$. For all other stages $h=2,...,K$, $n'^{(q)}_i=n^{(q)}_i - n^{(q-1)}_i$ patients are recruited to treatment $i=1,...,K$ in that stage; we use this superscript throughout to denote values that are computed based on the corresponding $n'^{(q)}_i$ for $i=1,...,K$ and $q=1,...,Q$ in the same way we described in Section~\ref{sec:msrec}. We assume that the proportion of patients recruited to each treatment remain constant over each stage of the trial.

For a vector of constants $\nu_{\mathcal{K}} = (\nu_{k_1},\nu_{k_2},...)$ where $(k_1,k_2,...)=\mathcal{K}$, each treatment $k\in\mathcal{K}$ and stage $q=1,...,Q$ we define the event,
\begin{equation*}
R_{\mathcal{K}}^{(q)}(\nu_k) = \{-C_{\mathcal{K},\alpha}^{(q)} + \Sigma_{\lambda=1}^{Q-1}w_k^{(\lambda,q)}(\theta_k - \nu_k)/\sigma_{p,k}'^{(\lambda)} < Z_k^{(q)} < C_{\mathcal{K},\alpha}^{(q)} + \Sigma_{\lambda=1}^{Q-1}w_k^{(\lambda,q)}(\theta_k - \nu_k)/\sigma_{p,k}'^{(\lambda)}\}
\end{equation*}
If $\theta_k = \nu_k$ for all $k=1,...,K(K-1)/2$ then $\bar{R}_{\mathcal{K}}(\nu_{\mathcal{K}}) = \bigcap_{k\in\mathcal{K}}\bigcap_{q=1}^{Q} R_{\mathcal{K}}^{(q)}(\nu_k)$
is the event that we fail to reject $H_{0,\mathcal{K}}$. Rearranging $R_{\mathcal{K}}^{(q)}(\nu_k)$ we have that $R_{\mathcal{K}}^{(q)}(\nu_k) = \{l_{k,\mathcal{K},\alpha}^{(q)} <  Z_k'^{(q)} < u_{k,\mathcal{K},\alpha}^{(q)}\}$
where $l_{k,\mathcal{K},\alpha}^{(q)}(\nu_k) = \{-C_{\mathcal{K},\alpha}^{(q)} - \Sigma_{\lambda=1}^{(Q-1)}w_k^{(\lambda,q)} t_k^{(q)} - w_k^{(\lambda,q)}\nu_k/\sigma_{p,k}'^{(\lambda)}\}/w_k^{(\lambda,q)}$ and $u_{k,\mathcal{K},\alpha}^{(q)}(\nu_k) = \{C_{\mathcal{K},\alpha}^{(q)} - \Sigma_{\lambda=1}^{(Q-1)}w_k^{(\lambda,q)} t_k^{(q)} - w_k^{(\lambda,q)}\nu_k/\sigma_{p,k}'^{(\lambda)}\}/w_k^{(\lambda,q)}$ with $t_k^{(q)} = (\hat{\theta}_k'^{(q)} - \theta)/\sigma_{p,k}'^{(q)}$. For each $q=1,...,Q-1$ $\boldsymbol{t}^{(q)} = (t_{k_1}^{(q)},t_{k_2}^{(q)},...)$ follows a multivariate normal distribution with a known correlation matrix, while $\boldsymbol{t}^{(1)},...,\boldsymbol{t}^{(Q-1)}$ are independent. Defining $\Omega_q$ as the sample space of $\boldsymbol{t}^{(q)}$, $L_{\mathcal{K},\alpha}^{(q)}(\nu_{\mathcal{K}})= \{l_{{k_1},\mathcal{K},\alpha}^{(q)}(\nu_{k_1}),l_{k_2,\mathcal{K},\alpha}^{(q)}(\nu_{k_2}),...\}$, $U_{\mathcal{K},\alpha}^{(q)}(\nu_{\mathcal{K}})=\{u_{{k_1},\mathcal{K},\alpha}^{(q)}(\nu_{k_1}),u_{k_2,\mathcal{K},\alpha}^{(q)}(\nu_{k_2}),...\}$ and $\Sigma'^{(q)}$ as the correlation matrix of $Z'^{(q)}_k$ for $q=1,...,Q$, and with $\Phi\{L_{\mathcal{K},\alpha}^{(q)}(\nu_{\mathcal{K}}),U_{\mathcal{K},\alpha}^{(q)}(\nu_{\mathcal{K}}),\Sigma'^{(q)}\}$ denoting the result of integrating the normal density function with mean zero and correlation matrix $\Sigma'^{(q)}$ over a region defined by lower limits $L_{\mathcal{K},\alpha}^{(q)}(\nu_{\mathcal{K}})$ and upper limits $U_{\mathcal{K},\alpha}^{(q)}(\nu_{\mathcal{K}})$, 
\begin{equation}
\begin{split}
P\{\bar{R}_{\mathcal{K}}(\nu_{\mathcal{K}})\} &= E(...E[P\{\bar{R}_{\mathcal{K}}(\nu_{\mathcal{K}})\}|\boldsymbol{t}^{(q)}]...|\boldsymbol{t}^{(1)}) \\
&= \int_{\Omega_{1}}...\int_{\Omega_{Q}} \Pi_{q=1}^Q\Phi\{L_{\mathcal{K}\alpha}^{(q)}(\nu_{\mathcal{K}}),U_{\mathcal{K},\alpha}^{(q)}(\nu_{\mathcal{K}}),\Sigma'^{(q)}\}\mathrm{d}\boldsymbol{t}^{(q)}...\mathrm{d}\boldsymbol{t}^{(1)}.\\
\end{split}
\label{eq:int}
\end{equation}

Equation~\ref{eq:int} can be evaluated numerically for all choices of the design parameters ($n_k^{(q)}$ and $C_{\mathcal{K},\alpha}$). Note that if $\nu_{\mathcal{K}} = \boldsymbol{0}$ this only depends on the ratios of the $n_k^{(q)}$. The consequence of this is that we may define testing structure based on these ratios where $r_i^{(q)}$ are used as described in Section~\ref{sec:masize}.

\cite{magirr2012generalized} discuss several familiar group sequential options \citep{jennison1999group,pocock1977group,o1979multiple,whitehead1997design,gordon1983discrete} which we may apply in this setting under the right constraints. Of particular interest is the $\alpha$-spending approach \citep{gordon1983discrete}, using Equation~\ref{eq:int} to iteratively compute  $C_{\mathcal{K},\alpha}^{(q)}$. Suppose we have a pre-defined monotonically increasing $\alpha$-spending function, $\alpha^\star(\tau)$, with $\alpha^\star(0)=0$ and $\alpha^\star(1)=\alpha$. As noted by \cite{magirr2012generalized} this requires some definition of information time, $0<\tau<1$, in the multi-arm setting \citep{follmann1994monitoring}. At interim analysis $\lambda=1,...,Q$ we define the error rate as $\alpha^{(q)} = \alpha^\star(\tau_q)$ and given $C_{\mathcal{K},\alpha}^{(q)}$ for $q=1,...,\lambda$ choose $C_{\mathcal{K},\alpha}^{(q)}$ such that
\begin{equation}
\label{eq:er}
\alpha^{(q)} = 1 - \int_{\Omega_{1}}...\int_{\Omega_{\lambda}} \Pi_{q=1}^{\lambda}\Phi\{L_{\mathcal{K},\alpha}^{(q)}(\boldsymbol{0}),U_{\mathcal{K},\alpha}^{(q)}(\boldsymbol{0}),\Sigma'^{(q)}\}\mathrm{d}\boldsymbol{t}^{(q)}...\mathrm{d}\boldsymbol{t}^{(1)}.
\end{equation}

For any treatment $i=1,...,K$ should all corresponding $H_{0,k}$ be globally rejected (for $k=1,...,K(K-1)/2$) it may be removed from the trial. Such behaviour would reduce the expected sample size of the trial as is typical for group sequential designs \citep{jennison1999group}. 

\subsection{Generalised all pairwise Dunnett test}

\begin{corollary}
\label{cor:msconsonance}
For every $\mathcal{K} \subseteq \{1,...,K(K-1)/2\}$ let the test of the corresponding intersection hypothesis $H_{0,\mathcal{K}} = \bigcap_{k\in \mathcal{K}} H_{0,k}$ be constructed using the group sequential all pairwise Dunnett test using an error spending function $\alpha^\star(\tau)$. This testing procedure is consonant. That is for $\mathcal{K}_1,\mathcal{K}_2\subseteq \{1,...,K(K-1)/2\}$ with corresponding intersections $H_{0\mathcal{K}_1} = \bigcap_{k\in\mathcal{K}_1}H_{0k}$ and  $H_{0\mathcal{K}_2} = \bigcap_{k\in\mathcal{K}_2}H_{0k}$, if $\mathcal{K}_1 \subset \mathcal{K}_2$ then $C_{\mathcal{K}_1,\alpha}^{(q)} < C_{\mathcal{K}_2,\alpha}^{(q)}$ for all $q=1,...,Q$.
\end{corollary}

\begin{definition}
For the global intersection hypothesis $H_{0,F} = \bigcap_{k=1}^{K(K-1)/2}H_{0k}$ we define testing boundaries $C_{\mathcal{F},\alpha} = (C_{\mathcal{F},\alpha}^{(1)},...,C_{\mathcal{F},\alpha}^{(q)})$. Under the generalised all pairwise Dunnett test we reject $H_{0k}:\theta_k = 0$ at analysis $q$ if $|Z_{k}^{(q)}| > C_{\mathcal{F},\alpha}^{(q)}$ for all $k=1,...,K(K-1)/2$.
\end{definition}

\begin{remark}
By Corollary~\ref{cor:msconsonance} $C_{\mathcal{F},\alpha}^{(q)}$ is conservative for all $\mathcal{K} \subset F$. Thus for the generalised all pairwise Dunnett implicitly all intersection hypotheses are tested such that under $H_{0,\mathcal{K}}$ $P(\text{Reject }H_{0,\mathcal{K}}) \leq \alpha$. Further to this if we reject $H_{0k}$ we simultaneously reject all associated $H_{0,\mathcal{K}}$ for which $k\in\mathcal{K}$. Through this the generalised all pairwise Dunnett test implies a closed testing procedure and thus strongly controls the Family-Wise Error Rate. 
\end{remark}

\subsection{Other results}

\begin{corollary}
\label{cor:lfc}
If $\sigma_i^2/n_i^{(1)} = \sigma_j^2/n_j^{(1)}$ for all $i,j=1,...,K$, the least favourable configuration minimizes $P(\text{Globally reject }H_{01})$. 
\end{corollary}

\begin{corollary}
\label{cor:errspend}
Suppose we use the same error spending function $\alpha^\star(\tau)$ to construct a group sequential all pairwise Dunnett test and the same error spending function to construct an equivalent group sequential test with Bonferroni adjustment (that is for each $k=1,...,K(K-1)/2$ $H_{0k}$ the group sequential test is defined as per the all pairwise Dunnet at $\alpha/{K(K-1)/2}$, or for each $q=1,...,Q$ $\alpha^{(q)}/{K(K-1)/2}$). The probability of globally rejecting $H_{0k}$ for $k=1,...,K(K-1)/2$ is higher for the group sequential all pairwise Dunnett testing procedure.
\end{corollary}

\section{Fully flexible multi-stage designs}

An alternative to group sequential designs is to define fully flexible multi-stage designs \citep{bretz2006confirmatory,schmidli2006confirmatory}. As described in Section~\ref{sec:mstest} both the cumulative and stage-wise z-statistics have known multivariate normal distributions. 

Suppose at each analysis $q=1,...,Q-1$ of the trial we wish to make a decision about how the next stage (or all remaining stages) of the trial will run, based on the accumulated observations up to stage $q$. For example, we may base the decision on the $Z_k^{(q)}$ for $k=1,...,K(K-1)/2$. Noting that the stage-wise z-statistics $Z_k'^{(q)}$ are conditionally independent for $q=1,...,Q$ we can use combination tests \citep{bauer1994evaluation} to combine the information across all stages of the trial while allowing any decision making framework to be applied. In most circumstances we expect the choice of the weighted inverse normal \citep{lehmacher1999adaptive} to be a strong choice but we give the following general construction. 

To apply the combination test we convert all the relevant test statistics to p-values ($Unif(0,1)$ random variables under the null hypothesis). For each $k=1,...,K(K-1)/2$ and $q=1,...,Q$ we have that
\begin{equation*}
P_k'^{(q)} = 1 - \Phi(Z_k'^{(q)}).
\end{equation*}
While for all other $\mathcal{K} \subset (1,...,K(K-1)/2)$, with $\mathcal{K} = (k_1,K_2,...)$ we define stage-wise test statistics for each stage $q=1,...,Q$
\begin{equation*}
Z_{\mathcal{K}}'^{(q)} = max(|Z_{k_1}'^{(q)}|,|Z_{k_2}'^{(q)}|,...),
\end{equation*}
from which we find the corresponding p-value by evaluating the following under $H_{0,\mathcal{K}}$
\begin{equation*}
P_{\mathcal{K}}'^{(q)} = 1 - P_{\theta_{\mathcal{K}}=\boldsymbol{0}}(\bigcap_{k\in\mathcal{K}} z_{\mathcal{K}}'^{(q)} < Z_k'^{(q)} < z_{\mathcal{K}}'^{(q)}).
\end{equation*}
Recall that under $H_{0,\mathcal{K}}$ $(Z_{k_1}'^{(q)},Z_{k_2}'^{(q)},...)$ follow a multivariate normal distribution with mean $0$ and known correlation matrix as described in Section~\ref{sec:msrec}. To test each null hypothesis at the end of the trial we use a combination function $\psi(.)$ to find combined p-values 
\begin{equation*}
P_{\mathcal{K}}^{(Q)} = \psi(P_{\mathcal{K}}'^{(1)},...,P_{\mathcal{K}}'^{(q)})
\end{equation*}
for all $\mathcal{K} \subset (1,...,K(K-1)/2)$. Noting that $P_{\mathcal{K}}'^{(1)},...,P_{\mathcal{K}}'^{(q)}$ are p-clud \citep{brannath2002recursive} $\psi(.)$ is chosen such that under $H_{0,\mathcal{K}}$ we have that $P_{\mathcal{K}}'^{(q)}  \sim Unif(0,1)$. Thus we may reject $H_{0,\mathcal{K}}$ at level $\alpha$ when $P_{\mathcal{K}}^{(Q)}<\alpha$.

\begin{remark}
This procedure will strongly control the FWER as all tests may be constructed at level $\alpha$. One could of course find the combined z-statistics for the trial and from these construct the tests for each intersection hypothesis based on this. The performance of the fully flexible all pairwise Dunnett under either choice and comparison with the alternatives we have discussed in the other Sections and will depend on the choice of the combination function. It should however be clear that one would expect a benefit over a Bonferroni adjustment and gatekeeping procedure as argued previously.
\end{remark}

\section{Discussion}

In this manuscript, we have introduced a Dunnett type method for making all pairwise comparisons in experiments with no control. We have shown that this method can be constructed to ensure that the FWER is strongly controlled at the pre-determined level. This allows for superior performance when compared to a Bonferroni type adjustment or gatekeeping procedure. We have shown how to apply this testing methodology in a variety of scenarios. Allowing its application in all forms of adaptive trial design for clinical trials, which would be a common setting for such methods. 

Through the least favourable configuration, we enable the design of such methods in a worst-case scenario. In practice, this may not be the configuration under which the trial is designed. With the ratios for recruitment chosen the rejection regions will not depend on the actual sample size allowing the finding of a suitable sample size under any given configuration of the treatment effect. As alternatives, one might consider powering the trial such that only one of the treatments has a different mean than the others. Alternatively, a setting, such as discussed in \cite{whitehead2020estimation}, where $T_j$ would be eliminated from the study with probability $\geq0.90$ if it is inferior to $T_i$ by a given margin could be considered. Whatever the choice the least favourable configuration then serves as a useful tool to understand the operating characteristics of the design and will be useful for comparing potential performance at the design stage.

Our work here has focussed on normally distributed observations with known variance, though this is not required. The original work of Dunnett is framed using the multi-variate t-distribution allowing for the incorporation of unknown variances. Likewise, the incorporation of other endpoints that are not normally distributed, such as ordinal outcomes or survival times is possible. Such additions are discussed by \cite{jaki2013considerations} in the context of the generalised Dunnett methodology. 

The interpretation of the results from such trials is less clear than in a typical multi-arm trial with a common control. While it will also be important to consider how early stopping of recruitment to a treatment (both for futility and efficacy) might be incorporated both from a practical and statistical viewpoint. In addition to the increased technical complexity of the design of such trials will also be computationally intensive. We shall address such discussion in future work to aid the design and analysis of such methods.

\section*{Acknowledgements}

This research was supported by the NIHR Cambridge Biomedical Research Centre (BRC-1215-20014). T Jaki received funding from the UK Medical Research Council (MC\_UU\_00002/14, MC\_UU\_00040/03). The views expressed in this publication are those of the authors. They are not necessarily those of the NHS, the National Institute for Health Research or the Department of Health and Social Care (DHSC). The funders and associated partners are not responsible for any use that may be made of the information contained herein. For the purpose of open access, the author has applied a Creative Commons Attribution (CC BY) licence to any Author Accepted Manuscript version arising.

\section*{Supplementary Material}

\begin{proof}[Proof of Lemma~\ref{lem:consonance}]
Let $\theta_{\mathcal{K}_1}$ and $\theta_{\mathcal{K}_2}$ be the vectors of $\theta_k$ for $k \in \mathcal{K}_1$ and $k \in \mathcal{K}_2$ respectively. Defining $Z_{\mathcal{K}_1}$ to be the maximum of the $|Z_k|$ for which $k \in \mathcal{K}_1$ and $Z_{\mathcal{K}_2}$ to be the maximum of the  $|Z_k|$ for which $k \in \mathcal{K}_2$. Under $H_{0\mathcal{K}_1}$ we have  $\theta_{\mathcal{K}_1} = \boldsymbol{0}$ and by construction
\begin{equation*}
P_{\theta_{\mathcal{K}_1} = \boldsymbol{0}}(Z_{\mathcal{K}_1} > C_{\mathcal{K}_1,\alpha}) = \alpha.
\end{equation*}
Similarly under $H_{0\mathcal{K}_2}$ we have  $\theta_{\mathcal{K}_2} = \boldsymbol{0}$ and by construction
\begin{equation*}
P_{\theta_{\mathcal{K}_2} = \boldsymbol{0}}(Z_{\mathcal{K}_2} > C_{\mathcal{K}_2,\alpha}) = \alpha.
\end{equation*}
Now consider the probability that $Z_{\mathcal{K}_2} > C_{\mathcal{K}_1,\alpha}$ under $H_{0\mathcal{K}_2}$,
\begin{equation*}
\begin{split}
P_{\theta_{\mathcal{K}_2} = \boldsymbol{0}}&(Z_{\mathcal{K}_2} > C_{\mathcal{K}_1,\alpha}) \\
&= P_{\theta_{\mathcal{K}_2} = \boldsymbol{0}}[\{{(\mathcal{K}_2} > C_{\mathcal{K}_1,\alpha})\cap (Z_{\mathcal{K}_1} = Z_{\mathcal{K}_2})\}\cup\{(Z_{\mathcal{K}_2} > C_{\mathcal{K}_1,\alpha})\cap (Z_{\mathcal{K}_1} < Z_{\mathcal{K}_2})\}]\\
&= P_{\theta_{\mathcal{K}_2} = \boldsymbol{0}}\{(Z_{\mathcal{K}_2} > C_{\mathcal{K}_1,\alpha})\cap (Z_{\mathcal{K}_1} = Z_{\mathcal{K}_2})\} + \{(Z_{\mathcal{K}_2} > C_{\mathcal{K}_1,\alpha})\cap (Z_{\mathcal{K}_1} < Z_{\mathcal{K}_2})\}\\
&= P_{\theta_{\mathcal{K}_2} = \boldsymbol{0}}(Z_{\mathcal{K}_1} > C_{\mathcal{K}_1,\alpha}) + P_{\theta_{\mathcal{K}_2} = \boldsymbol{0}}(Z_{\mathcal{K}_1} < C_{\mathcal{K}_1,\alpha})P_{\theta_{\mathcal{K}_2} = \boldsymbol{0}}(Z_{\mathcal{K}_2} > C_{\mathcal{K}_1,\alpha}|Z_{\mathcal{K}_2} > C_{\mathcal{K}_1,\alpha}).
\end{split}
\end{equation*}
Since  $\mathcal{K}_1 \subset \mathcal{K}_2$ we have that $\theta_{\mathcal{K}_2} = \boldsymbol{0}$ implies $\theta_{\mathcal{K}_1} = \boldsymbol{0}$ thus under $H_{0\mathcal{K}_2}$ 
\begin{equation*}
P_{\theta_{\mathcal{K}_2} = \boldsymbol{0}}(Z_{\mathcal{K}_1} > C_{\mathcal{K}_1,\alpha}) = \alpha
\end{equation*} 
and 
\begin{equation*}
P_{\theta_{\mathcal{K}_2} = \boldsymbol{0}}(Z_{\mathcal{K}_1} < C_{\mathcal{K}_1,\alpha}) = 1-\alpha.
\end{equation*} 
Note also that under $H_{0\mathcal{K}_2}$
\begin{equation*}
P_{\theta_{\mathcal{K}_2} = \boldsymbol{0}}(Z_{\mathcal{K}_2} > C_{\mathcal{K}_1,\alpha}|Z_{\mathcal{K}_1} < C_{\mathcal{K}_1,\alpha}) > 0
\end{equation*}
and so 
\begin{equation*}
P_{\theta_{\mathcal{K}_2} = \boldsymbol{0}}(Z_{\mathcal{K}_2} > C_{\mathcal{K}_1,\alpha}) > P_{\theta_{\mathcal{K}_2} = \boldsymbol{0}}(Z_{\mathcal{K}_2} > C_{\mathcal{K}_2,\alpha}).
\end{equation*}
Thus we have $C_{\mathcal{K}_1,\alpha} < C_{\mathcal{K}_2,\alpha}$ as required.
\end{proof}

\begin{proof}[Proof of Corollary~\ref{cor:msconsonance}]
For any given $\mathcal{K}_1,\mathcal{K}_2 \subset (1,...,K(K-1)/2)$ at stage $q=1,...,Q$ the error spending function dictates that the test be constructed at level $\alpha^{(q)}$. For $q=1$ it is clear that this is consonant by Lemma~\ref{lem:consonance}. 
While for $h>1$ the error from the previous $Q-1$ stages is fixed at $\alpha^{(Q-1)}$ with $C_{\mathcal{K}_1,\alpha}^{(q)}$ and $C_{\mathcal{K}_2,\alpha}^{(q)}$ only contributing to Equation~\ref{eq:er} 
through $\Phi\{L_{\mathcal{K},\alpha}^{(q)}(\boldsymbol{0}),U_{\mathcal{K},\alpha}^{(q)}(\boldsymbol{0}),\Sigma'^{(q)}\}$. Therefore $\Phi\{L_{\mathcal{K},\alpha}^{(q)}(\boldsymbol{0}),U_{\mathcal{K},\alpha}^{(q)}(\boldsymbol{0}),\Sigma'^{(q)}\}$ must take the same value for both $C_{\mathcal{K}_1,\alpha}^{(q)}$ and $C_{\mathcal{K}_2,\alpha}^{(q)}$, clearly following the arguments from the proof of Lemma~\ref{lem:consonance} 
 we have that $C_{\mathcal{K}_1,\alpha}^{(q)} < C_{\mathcal{K}_2,\alpha}^{(q)}$. Thus we have the result for all $q=1,...,Q$.
\end{proof}

\begin{proof}[Proof of Corollary~\ref{cor:lfc}]
Consider the probability of rejecting $H_{01}$, this requires each associated $H_{0,\mathcal{K}}$ for which $k\in\mathcal{K}$ to have been rejected at some stage. That is for some $q=1,...,Q$, $Z_{\mathcal{K}}^{(q)} > C_{\mathcal{K},\alpha}^{(q)}$. The probability of each of these possible events is directly comparable to the single-stage all pairwise Dunnett test, and thus the result follows directly from the proof of Theorem~\ref{the:lfc}.
\end{proof}

\begin{proof}[Proof of Corollary~\ref{cor:errspend}]
Consider at any given stage $q=1,...,Q$ the group sequential all pairwise Dunnett is constructed such that each hypothesis is tested at $\alpha^{(q)}$. This result follows directly from Theorem~\ref{the:bonferonni}.
\end{proof}

\bibliographystyle{apalike}
\bibliography{APD}

\end{document}